\documentclass[conference,a4paper]{IEEEtran}

\usepackage{times}
\usepackage{algorithmic}
\usepackage{amsmath}
\usepackage{amssymb}
\usepackage{amsthm}
\usepackage{color}
\usepackage{colortbl}
\usepackage{float}
\usepackage{cite}
\usepackage{subfig}
\usepackage{verbatim}
\usepackage[pdftex]{graphicx}

\newtheorem{theorem}{Theorem}
\newtheorem{example}[theorem]{Example}
\newtheorem{definition}[theorem]{Definition}
\newtheorem{remark}[theorem]{Remark}

\newcommand{\beq}{\begin{eqnarray}}
\newcommand{\eeq}{\end{eqnarray}}

\newcommand{\field}[1]{\mathbb{#1}}

\newcommand{\F}{\field{F}}

\newcommand{\cM}{{\cal M}}
\newcommand{\cC}{{\cal C}}

\newfont{\bbb}{msbm10 scaled 500}

\newfont{\bb}{msbm10 scaled 1100}

\newcommand{\cv}{{\bf c}}

\newcommand{\fv}{{\bf f}}

\newcommand{\uv}{{\bf u}}

\newcommand{\vv}{{\bf v}}
\newcommand{\xv}{{\bf x}}

\newcommand{\Um}{{\bf U}}

\newcommand{\Vm}{{\bf V}}

\newcommand{\Ac}{{\cal A}}

\newcommand{\Mc}{{\cal M}}

\newcommand{\gammav}{\hbox{\boldmath$\gamma$}}

\definecolor{OXO-emph}{RGB}{153,0,0}

\newcommand\ceilb[1]{\left\lceil #1 \right\rceil}

\begin{document}

\sloppy

\title{Optimal Locally Repairable Codes\\ via Rank Metric Codes}
%
 \author{
   \IEEEauthorblockN{
   Natalia~Silberstein, Ankit~Singh~Rawat, O.~Ozan~Koyluoglu, and~Sriram~Vishwanath}
   \IEEEauthorblockA{Department of Electrical and Computer Engineering\\
     The University of Texas at Austin\\
     TX 78712 USA\\
     Email:
     \{natalys, ankitsr, ozan, sriram\}@austin.utexas.edu.}
 }
\maketitle


\begin{abstract}

This paper  presents a new explicit construction for  locally repairable codes (LRCs) for distributed storage systems which possess all-symbols locality and  maximal possible minimum distance, or equivalently, can tolerate the maximal number of node failures.  This construction,  based on maximum rank distance (MRD) Gabidulin codes,
 provides new optimal vector and scalar LRCs. In addition, the paper also discusses mechanisms by which codes obtained using this construction can be used to construct  LRCs with efficient repair of failed nodes by combination of LRC with regenerating codes.
\end{abstract}


\IEEEpeerreviewmaketitle


\section{Introduction}

In distributed storage systems (DSS), it is desirable that data be reliably stored over a network of nodes in such a way that a user (\emph{data collector}) can retrieve the stored data even if some nodes fail. To achieve such a resilience against node failures, DSS introduce data redundancy based on different coding techniques. For example, erasures codes are widely used in such systems: When using an $(n,k)$ code, data to be stored is  first divided into $k$ blocks; subsequently, these $k$ information blocks are encoded into $n$ blocks stored on $n$ distinct nodes in the system.  In addition, when a single node fails, the system reconstructs the data stored in the failed node to keep the required level of redundancy.  This process of data reconstruction for a failed node is called \emph{node repair process}~\cite{dimakis}. During a node repair process,  the node which is added to the system to replace the failed node downloads  data from a set of appropriate and accessible nodes.

There are two important goals that guide the design of codes for DSS: reducing the \emph{repair bandwidth}, i.e. the amount of data downloaded from system nodes during the node repair process,  and achieving \emph{locality}, i.e. reducing the number of nodes participating in the node repair process. These goals underpin the design of two families of codes for DSS called \emph{regenerating codes} (see~\cite{dimakis,wu3,shah,suh1,rashmi, bruck11,DRWS11,DaOg11} and references therein) and \emph{locally repairable codes} (see~\cite{gopalan, oggier_hom, oggier_proj, prakash, dimitris, HSXOCGLY12, SAPDVCB13, pyramid, HL-M2007,Kumar12,Hollman13}), respectively.

In this paper we focus on the locally repairable codes (LRCs). Recently, these codes have drawn significant attention within the research community.
Oggier et al.~\cite{oggier_hom, oggier_proj} presents coding schemes which facilitate local node repair. In~\cite{gopalan}, Gopalan et al. establishes an upper bound on the minimum distance of  scalar LRCs, which is analogous to the Singleton bound. The paper also showes that pyramid codes, presented in \cite{pyramid}, achieve this bound with information symbols locality. Subsequently, the work by Prakash et al. extends the bound to a more general definition of scalar LRCs~\cite{prakash}. (Han and Lastras-Montano~\cite{HL-M2007} provide a similar upper bound which is coincident  with the one in~\cite{prakash} for small minimum distances, and also present codes that attain this bound in the context of reliable memories.) In \cite{dimitris}, Papailiopoulos  and Dimakis generalize the bound in \cite{gopalan} to vector codes,  and present locally repairable coding schemes which exhibits MDS property at the cost of small amount of additional storage per node.

The main contributions of our paper are as follows. First, in Section~\ref{sec:preliminaries}, we generalize the definition of \emph{scalar} locally repairable codes, presented in~\cite{prakash} to \emph{vector} locally repairable codes. For such codes, every node storing $\alpha$ symbols from a given field $\F$, can be locally repaired by using data stored in at most $r$ other nodes from a group of nodes  of size $r+\delta-1<n$, which we call a \emph{local group}, where $n$ is the number of system nodes, and $r$ and $\delta$ are the given locality parameters.  Subsequently, in Section~\ref{sec:vectorLRC}, we derive an upper bound on the minimum distance $d_{\min}$ of the vector codes that satisfy a given locality constraint, which establishes a trade off between node failure resilience (i.e., $d_{\min}$) and per node storage $\alpha$. \footnote{In a parallel and independent work,\cite{Kumar12}, Kamath et al. also provide upper bounds on minimum distance together with constructions and existence results for vector LRCs.}
The bound presented in~\cite{dimitris} can be considered as a special case of our bound with $\delta=2$.
Further, we present an explicit construction for LRCs which attain this bound on  minimum distance.
This construction is based on maximum rank distance (MRD) Gabidulin codes, which are a rank-metric analog of Reed-Solomon codes.
The \emph{scalar} and \emph{vector} LRCs that are obtained by this construction are the first explicit optimal locally repairable codes with $(r+\delta-1)\nmid n$. Finally, in Section~\ref{sec:discussion}, we discuss how the scalar and vector codes obtained by this construction can be used for constructions of repair bandwidth efficient  LRCs. We conclude the paper with Section~\ref{sec:conclusions}.

\section{Background}\label{sec:preliminaries}

\subsection{System Parameters}

Let $\cM$ be the size of a file $\fv$ to be stored in a DSS with $n$ nodes.
All data symbols belong to a finite field $\F$. Each node stores
$\alpha$ symbols over $\F$.

\subsection{Locally Repairable Codes}
We generalize the definition of \emph{scalar} locally repairable codes, presented in~\cite{prakash} to \emph{vector} locally repairable codes, where each node $i$, $1\leq i\leq n$ stores a vector $\xv_i$ of length $\alpha$ over $\F$.

First, we provide an alternate definition of the minimum distance of a vector code~\cite{gopalan,dimitris}.

\begin{definition}\label{def:dmin}
The minimum distance $d_{\min}$ of a vector code $C$ of dimension $\cM$ is defined as
\begin{equation}
d_{\min} =  n - \max_{\Ac \subseteq [n]: H(\xv_{\Ac}) < \mathcal{M}}|\Ac|,
\end{equation}
where $\Ac = \{i_1,\ldots, i_{|\Ac|}\} \subseteq [n]$ and $\xv_{\Ac} = (\xv_{i_1},\ldots, \xv_{i_{|\Ac|}})$.
\end{definition}


It follows from the definition of $d_{\min}$ that  the system can tolerate any $d_{\min}-1$ node failures, or equivalently, a data collector can reconstruct the original data $\fv$ by contacting any set of $n-d_{\min}+1$ storage nodes in the DSS. We are interested in ensuring this property of the DSS for its entire life span through the course of multiple failures and repairs.

\begin{definition}
\label{def:vectorLRC}
We say that a vector code $C$ has
$(r,\delta)$ \emph{locality} if for each node $i$, $1\leq i\leq n$, storing vector $\xv_i$ (of length $\alpha$),
there exists a set of nodes $\Gamma(i)$ such that
\begin{itemize}
\item $i\in\Gamma(i)$
\item $|\Gamma(i)|\leq r+\delta - 1$
\item Minimum distance of $C|_{\Gamma(i)}$
is at least $\delta$.
\end{itemize}

Note that the last two properties imply that each element $j \in \Gamma(i)$ can be written as a function of a set of at most $r$ elements in $\Gamma(i)$ (not containing $j$) and that $H(\Gamma(i))\leq r\alpha$.

Codes that satisfy these properties are called $(r, \delta, \alpha)$ \emph{locally repairable codes} (LRCs).
\end{definition}

Note, that these codes are generalizations of vector LRCs given in~\cite{dimitris}, which considered only the $\delta=2$ case.

\begin{remark}  $(r,\delta,1)$ locally repairable codes are named as scalar $(r,\delta)$ locally repairable codes.
\end{remark}

Prakash et al.~\cite{prakash} provided the following upper bound on the minimum distance of an $(r,\delta,1)$ LRC:
\begin{equation}
\label{scalarUpperBound}
d_{\min}\leq n-\cM+1-\left(\left\lceil\frac{\cM}{r}\right\rceil-1\right)(\delta-1).
\end{equation}

It was established in~\cite{prakash} that a family of  pyramid codes, presented in~\cite{pyramid} attains this bound and has \emph{information locality}, i.e. only information symbols satisfy the locality constraint.
However, an explicit construction of optimal scalar LRCs with \emph{all-symbols locality} is known only for the case $n=\left\lceil\frac{M}{r}\right\rceil(r+\delta-1)$~\cite{prakash,HL-M2007}.
The \emph{existence} of optimal scalar codes with all-symbols locality is shown for the case when $(r+\delta-1)|n$ and field size $|\F|>\cM n^{\cM}$~\cite{prakash}.
In this paper, we provide an explicit construction of optimal scalar LRCs  with all-symbols locality without the restriction $(r+\delta-1)|n$.

The following upper bound on the minimum distance of $(r,2,\alpha)$ LRCs  and construction of codes that attain this bound was presented in~\cite{dimitris}:
\begin{equation}
\label{eq:dimitrisBound}
d_{\min}\leq n-\left\lceil\frac{\cM}{\alpha}\right\rceil-\left\lceil\frac{\cM}{r\alpha}\right\rceil+2
\end{equation}

In the sequel, we generalize this bound for any $\delta\geq 2$ and present $(r,\delta,\alpha)$  LRCs that attain this bound.

\subsection{Maximum Rank Distance (MRD) Codes}

For the construction
presented in this paper we propose a precoding of
the file with maximum rank distance codes~\cite{Gab85,Rot91}.

Let $\F_{q^m}$ be an extension field of $\F_q$. Since $\F_{q^m}$
can be also considered as an $m$-dimensional vector space over $\F_q$, any element  $\gamma\in\F_{q^m}$ can be represented  as the vector ${\gammav=(\gamma_1,\ldots, \gamma_m)\in \F_q^m}$ , such that $\gamma=\sum_{i=1}^{m}b_i\gamma_i$, for a fixed basis $\{b_1,\ldots, b_m\}$ of the field extension. Similarly, any vector $\bf{v}$ $=(v_1, \ldots, v_N)\in\F_{q^m}^N$ can be represented by an $m\times N$ matrix $\Vm=[v_{i,j}]$ over $\F_q$, where each entry $v_i$  of $\vv$ is expanded as a column vector $(v_{i,1},\ldots,v_{i,m})^T$.
\begin{definition}
The \emph{rank} of a vector $\vv\in\F_{q^m}^N$, denoted by $\rm{rank}(\vv)$, is defined as the rank of the $m\times N$ matrix $\Vm$ over $\F_q$. Similarly, for two vectors $\vv,\uv \in \F_{q^m}^N$, the \emph{rank distance} is defined by $d_R(\vv,\uv)=\rm{rank}(\Vm-\Um)$.
\end{definition}

An $[N,K,D]_{q^m}$ \textmd{rank-metric code} $\cC\subseteq\F_{q^m}^N$ is a linear block code over $\F_{q^m}$ of length $N$, dimension $K$  and minimum rank distance $D$. A rank-metric code that attains the Singleton bound $D\leq N-K+1$ in rank-metric is called \emph{maximum rank distance} (MRD) code.
For $m\geq N$, a construction of MRD codes was presented by Gabidulin~\cite{Gab85}.
In the similar way as Reed-Solomon codes, Gabidulin codes can be obtained by
evaluation of polynomials, however, for Gabidulin codes the special family of polynomials,
called \emph{linearized polynomials}, is used:

\begin{definition}
A linearized polynomial $f(x)$ over $\F_{q^m}$ of $q-$degree $t$ has the form $f(x)=\sum_{i=0}^{t}a_ix^{q^i}$,
where $a_i\in \F_{q^m}$, and $a_{t}\neq 0$.
\end{definition}
 Note, that evaluation of a linearized polynomial is an $\F_{q}-$linear
transformation from $\F_{q^m}$ to itself, i.e., for any $a, b \in \F_q$ and
$\gamma_1, \gamma_2\in\F_{q^m}$, we have $f(a\gamma_1 + b\gamma_2)= af(\gamma_1)+ bf(\gamma_2)$~\cite{MWSl78}.

A codeword in an $[N,K,D=N-K+1]_{q^m}$ Gabidulin code
$\cC^{\rm{Gab}}$, $m\geq N$, is defined as ${\mathbf{c} = (f(g_1),f(g_2),\ldots, f(g_N))\in\F_{q^m}^N}$, where $f(x)$ is a
linearized polynomial over $\F_{q^m}$ of $q-$degree $K-1$ with the coefficients given by the information message, and $g_1,\ldots, g_N\in \F_{q^m}$
are linearly independent over $\F_q$~\cite{Gab85}.

An MRD code $\cC^{Gab}$ with minimum distance $D$ can correct any $D-1=N-K$  erasures, which we will call \emph{rank erasures}.  An algorithm for erasures correction of Gabidulin codes can be found e.g. in~\cite{GaPi08}.
%

\subsection{MDS Array Codes}

\begin{definition}
A linear $[\alpha\times n, k, d_{\min}]$ \textit{array code} $C$ of dimensions $\alpha\times n$ over $\F_q$ is defined as a linear subspace of $\F_q^{\alpha n}$. Its
minimum distance $d_{\min}$ is defined as the minimum Hamming distance over  $\F_{q}^\alpha$, when
we consider the codewords of $C$ as vectors of length $n$ over $\F_{q}^\alpha\cong \F_{q^{\alpha}}$. An array code $C$ is
called a \textit{maximum distance separable} (MDS) code if $|C|=q^{\alpha k}$, where $k=n-d_{\min}+1$.
\end{definition}

Constructions for MDS array codes can be  found e.g. in~\cite{blaum,BlRo99, CaBr06}.
Note, that an MRD code (in the matrix representation) is also an MDS array code.

\subsection{Regenerating Codes}

Regenerating codes are a family  of codes for distributed storage  that allow for efficient repair of failed nodes. When using such codes we assume that a data collector can reconstruct the original file by downloading the
data stored in any set of $k$ out of $n$ nodes.
When a node fails, its content can be reconstructed by downloading $\beta\leq \alpha$ symbols from
any $d$, $k\leq d\leq n-1$, surviving nodes.
Given a file size $\mathcal{M}$, a trade-off between storage per node $\alpha$ and {\em repair bandwidth} $\gamma \triangleq d\beta$ can be established.
 Two classes of codes that achieve two extreme points of this trade-off are known as {\em minimum storage regenerating (MSR)} codes and {\em minimum bandwidth regenerating (MBR)} codes. The parameters $(\alpha,\gamma)$ for MSR and MBR codes are given by $\left(\frac{\cM}{k},\frac{\cM d}{k(d-k+1)}\right)$, $\left(\frac{2\cM d}{2kd-k^2+k},\frac{2\cM d}{2kd-k^2+k}\right)$, respectively~\cite{dimakis}.


%


\section{Optimal Locally Repairable Codes}
\label{sec:vectorLRC}

In this section,
we first derive an upper bound on the minimum distance of $(r, \delta, \alpha)$ locally repairable codes.
Next, we propose a general code construction which attains the derived bound on $d_{\min}$.
Our approach is to apply  a two-stage encoding, where we use Gabidulin codes (a rank-metric analog of Reed-Solomon codes)
 along with MDS array codes. This construction can be viewed as a generalization of the construction proposed in~\cite{SRS12}.



\subsection{Upper Bound on $d_{\min}$ for an $(r, \delta, \alpha)$ LRC}
\label{subsec:d_min_upper}

We state a generic upper bound on the minimum distance $d_{\min}$ of an $(r, \delta, \alpha)$ code $C$. This bound generalizes the  bound given in \cite{dimitris} for LRCs with a single local parity ($\delta = 2$) to LRCs with multiple local parities ($\delta \geq 2$).
\begin{theorem}\label{thm:dmin}
Let $C$ be an $(r, \delta, \alpha)$ LRC. Then, it follows that
\begin{align}
\label{eq:upp_bound}
&d_{\min}(C) \leq n - \ceilb{\frac{\mathcal{M}}{\alpha}} + 1 - \left(\ceilb{\frac{\mathcal{M}}{r\alpha}}-1\right)(\delta - 1).
\end{align}
\end{theorem}
\begin{proof}
We follow the proof technique of \cite{gopalan, dimitris}. In particular, the proof involves construction of a set of nodes $\mathcal{A}$ for a locally repairable DSS such that total entropy of the symbols stored in $\mathcal{A}$ is less than $\Mc$ and
\begin{align}
\label{eq:setA}
|\mathcal{A}| \geq \ceilb{\frac{\mathcal{M}}{\alpha}} - 1 + \left(\ceilb{\frac{\mathcal{M}}{r\alpha}}-1\right)(\delta - 1).
\end{align}
Theorem~\ref{thm:dmin} then follows from Definition~\ref{def:dmin} and \eqref{eq:setA}. See~\cite{RKSV12} for the detailed proof.
\end{proof}

Remarkably, the theorem above establishes a trade-off between node failure resilience (i.e., $d_{\min}$) and per node storage $\alpha$, where $\alpha$ can be increased  to obtain higher $d_{\min}$. This is of particular interest in the design of codes having both locality and strong resilience to node failures.

\begin{remark}
For the special case of $\delta=2$, this bound matches with the bound~(\ref{eq:dimitrisBound}) presented in~\cite{dimitris}. For the case of $\alpha=1$, the bound reduces to $d_{\min} \leq n-\cM+1+(\ceilb{\cM/r}-1) (\delta-1)$, which is coincident with the bound~(\ref{scalarUpperBound}) presented in~\cite{prakash}.
\end{remark}


\subsection{Construction of $d_{\min}$-Optimal Vector LRCs}
\label{subsec:optimal_local_repairable}
In this subsection we present a construction of an $(r,\delta,\alpha)$ LRC which attains the bound given in Theorem~\ref{thm:dmin}.

\textbf{Construction I.} Consider a file $\fv$ over $\F=\F_{q^m}$ of size $\cM \geq r\alpha$. We encode the file in two steps before storing it on DSS. First, the file is encoded using a Gabidulin code. The codeword  of the Gabidulin code is then partitioned into local groups and each local group is then encoded using an MDS array code over $\mathbb{F}_q$.

In particular, let $\cM, r,\delta,\alpha, N, m$ be the positive integers such that $m\geq N\geq \cM\geq r\alpha$, and
 let $\cC^{\rm{Gab}}$ be an $[N, \cM, D = N-\cM+1]_{q^m}$ Gabidulin  code. We assume in this construction that $\alpha|N$. We denote by $g=\left\lceil\frac{N}{r\alpha}\right\rceil$ the number of local groups in the system.
\begin{itemize}
  \item If $\frac{N}{\alpha}\equiv 0(\textmd{mod } r)$ then a codeword $\cv \in \cC^{\rm{Gab}}$ is first partitioned into $g=\frac{N}{r\alpha}$ disjoint groups, each of size $r\alpha$, and each group is stored on a different set of $r$ nodes, $\alpha$ symbols per node. In other words, the output of the first encoding step generates the encoded data stored on $rg$ nodes, each one containing $\alpha$ symbols of a (folded) Gabidulin codeword. Second, we generate $\delta-1$ parity nodes per group by applying an $[\alpha\times (r+\delta-1),r,\delta]$ MDS array code over $\F_q$
      on each local group of $r$ nodes, treating these $r$ nodes as input data blocks (of length $\alpha$) for the MDS array code. At the end of the second round of encoding, we have $n=(r+\delta-1)g=\frac{N}{\alpha}+\frac{N}{r\alpha}(\delta-1)$ nodes, each storing $\alpha$ symbols over $\mathbb{F}_{q^{m}}$, partitioned into $g$ local groups, each of size $r-\delta+1$.

  \item If $\frac{N}{\alpha}\equiv \beta_0 (\textmd{mod } r)$, for $0 < \beta_0\leq r-1$, then a codeword $\cv \in \cC^{\rm{Gab}}$ is first partitioned into $g-1=\lfloor\frac{N}{r \alpha }\rfloor$ disjoint groups of size $r\alpha$ and one additional group of size $\beta_0\alpha$, the first $g-1$ groups are stored on $r(g-1)$ nodes, and the last group is stored on $\beta_0$ nodes, each one containing $\alpha$ symbols of a (folded) Gabidulin codeword. Second, we generate $\delta-1$ parity nodes per group by applying an $[\alpha\times (r+\delta-1),r,\delta]$ MDS array code over $\F_{q}$  on each of the first $g-1$ local groups of $r$ nodes and by applying a $[\alpha \times (\beta_0+\delta-1),\beta_0,\delta]$ MDS array code over $\F_{q}$ on the last local group.  At the end of the second round of encoding, we have $n=(r+\delta-1)(g-1)+(\beta_0+\delta-1)=\frac{N}{\alpha}+\left\lceil\frac{N}{r\alpha}\right\rceil(\delta-1)$ nodes, each storing $\alpha$ symbols over $\mathbb{F}_{q^{m}}$, partitioned into $g$ local groups, $g-1$ of which of size $r-\delta+1$ and one group of size $\beta_0+\delta-1$.
\end{itemize}

We denote the obtained code by $C^{\rm{loc}}$.

\begin{remark}
\label{rm:linearized property}
Note, that since an MDS array code from Construction I is defined over $\F_q$,
any symbol of any node of $C^{\rm{loc}}$ can be written as
$\sum_{j=1}^{r \alpha}a_j c_{i_j}=\sum_{j=1}^{r \alpha}a_j f(g_{i_j})=f(\sum_{j=1}^{r \alpha}a_j g_{i_j})$, where  $a_j\in \F_q$,  $c_{i_j}\in \F_{q^m}$ are $r \alpha$ symbols of the same group of $\cv$, and $g_{i_j}$, are linearly independent over $\F_q$  evaluation points.
Hence, \emph{any} $s\leq r \alpha$ symbols inside a group of $C^{\rm{loc}}$ are evaluations of $f(x)$ in $s$ linearly independent over $\F_q$ points. (If there is a group with $\beta_0<r$ elements we have the same result substituting $r$ with $\beta_0$). Thus, \emph{any} $\delta-1+i$ \emph{node} erasures in a group correspond to $i \alpha$ \emph{rank} erasures.
Moreover, if we take any $r \alpha$ symbols of $C^{\rm{loc}}$  from every group (and $\alpha \beta_0$ symbols from the smallest group, if it exists), we obtain a Gabidulin codeword, for a corresponding choice of evaluation points for a Gabidulin code, which encodes the given data $\cM$.

\end{remark}

\begin{theorem}
\label{trm:vector_optimality}
Let $C^{\rm{loc}}$ be the $(r, \delta, \alpha)$ locally repairable code $C^{\rm{loc}}$ obtained by Construction I. Then, for $N\geq \cM$, s.t. $\alpha|N$,
\begin{itemize}
  \item If $(r \alpha )|N$
  then $C^{\rm{loc}}$ attains the bound~(\ref{eq:upp_bound}).
  \item If  $\frac{N}{\alpha}(\textmd{mod } r)\geq \left\lceil\frac{\cM}{\alpha}\right\rceil(\textmd{mod } r)>0 $, then $C^{\rm{loc}}$ attains the bound~(\ref{eq:upp_bound}).
\end{itemize}

\end{theorem}

\begin{proof}
The proof  is based on Remark~\ref{rm:linearized property} and the observation that any
 $n - \ceilb{\frac{\mathcal{M}}{\alpha}} - \left(\ceilb{\frac{\mathcal{M}}{r\alpha}}-1\right)(\delta - 1)$ node erasures correspond to at most $D-1$ rank erasures which can be corrected by the Gabidulin code $\cC^{\rm{Gab}}$.
 See the details in Appendix~\ref{ap:appendix}.
  \end{proof}

Next, we write the conditions for the construction of $d_{\min}$-optimal code $C^{\rm{loc}}$
in terms of the given system parameters $n,\cM,r, \delta, \alpha$.

\begin{theorem}
\label{cor.parameters}
Let $C^{\rm{loc}}$ be a code obtained by Construction~I.
\begin{itemize}
  \item If $ (r+\delta-1)|n$, then $C^{\rm{loc}}$ is optimal, with the length of the corresponding Gabidulin code equal to  $N=\frac{n r \alpha}{r+\delta-1}$ and the field size $\F\geq q^N$.
  \item If $n(\textmd{mod } r+\delta-1)-(\delta-1)\geq
 \left\lceil\frac{\cM}{\alpha}\right\rceil(\textmd{mod } r)>0$
  then $C^{\rm{loc}}$ is optimal, with  the length of the corresponding Gabidulin code equal to $N=\alpha\left(n-\delta+1-(\delta-1)\left\lfloor\frac{n}{r+\delta-1}\right\rfloor\right)$ and the field size $\F\geq q^N$.
\end{itemize}

\end{theorem}

\begin{remark} For the case $\alpha=1$ Construction I provides $d_{\min}$-optimal scalar LRCs.
Note that this is the first explicit construction of optimal scalar locally repairable codes with $(r+\delta-1)\nmid n$.
\end{remark}

\begin{remark}
The required field size  $|\F|=q^m$ for the proposed construction should satisfy $m\geq N$, for any choice of $q$. So we can assume that $|\F|=q^N$, for $N$ given in Theorem~\ref{cor.parameters}. Note that
we can reduce the field size
 to $|\F|=q^{N/\alpha}$  by stacking~\cite{Kumar12} of $\alpha$ independent optimal scalar LRCs, obtained by Construction I.
 \end{remark}

We illustrate the construction of $C^{\rm{loc}}$ in the following examples. First we consider the scalar case.


\begin{example}
Consider the following system parameters:
$$(\cM,n,r,\delta,\alpha)=(9,14,4,2,1).
$$
Let $N=14-2+1-1\cdot \left\lfloor\frac{14}{4+2-1}\right\rfloor=11$.
First, $\cM=9$ symbols over $\F=\F_{2^{11}}$  are encoded into a codeword $\cv$ of a $[11,9,3]_{2^{11}}$ Gabidulin code $\cC^{\rm{Gab}}$. This codeword is partitioned into three groups, two of size $4$ and one of size $3$, as follows: $\cv=(a_1,a_2,a_3,a_4| b_1,b_2,b_3,b_4|c_1,c_2,c_3)$. Then, by applying a $[5,4,2]$ MDS code in the first two groups and a $[4,3,2]$ MDS code in the last group we add one parity to each group. The symbols of $\cv$ with three new parities $p_a,p_b,p_c$ are stored on 14 nodes as shown in Fig~\ref{fig:construction1}.
\begin{figure}[t]
 \centering
 \includegraphics[width=1\columnwidth]{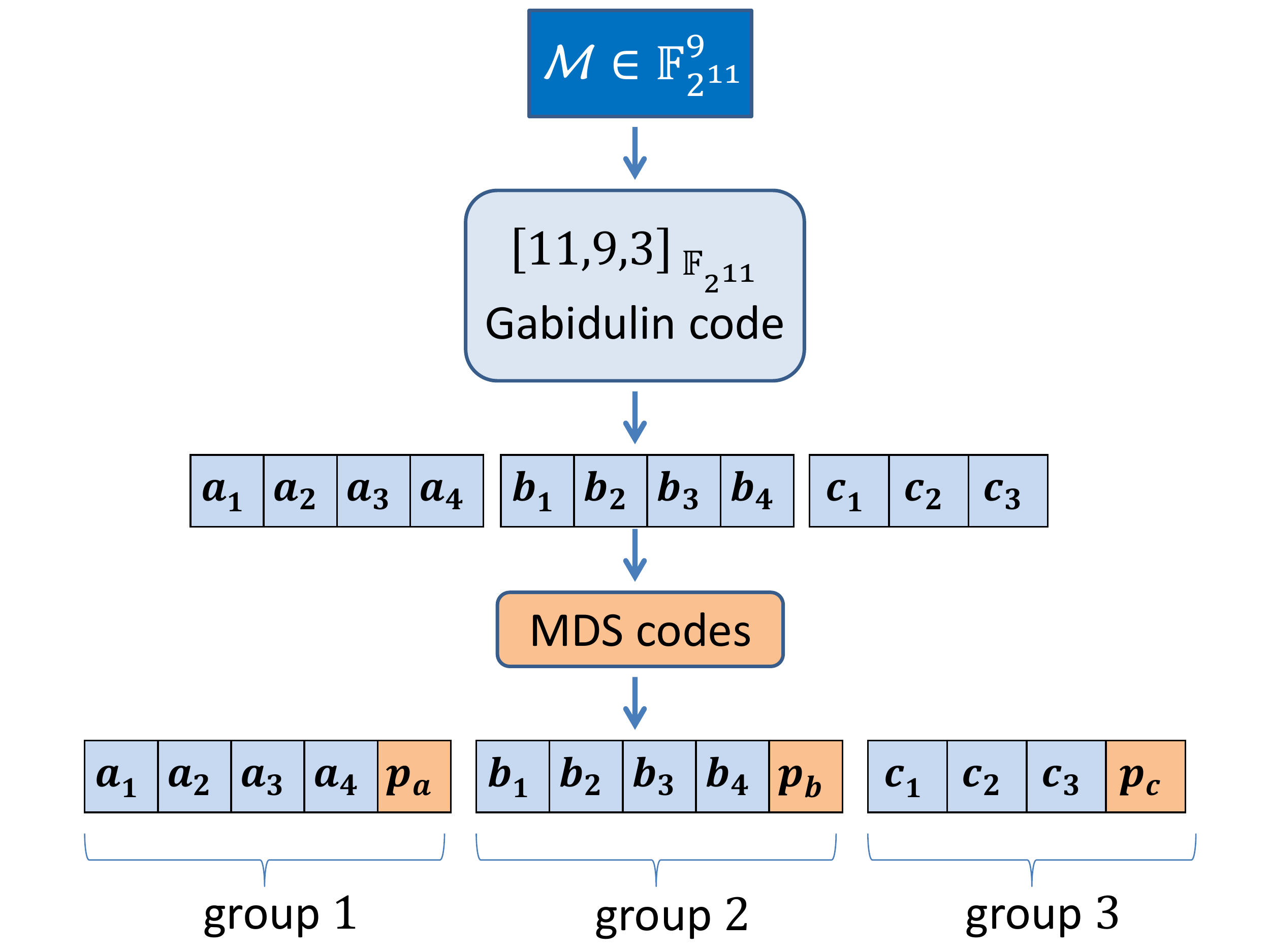}
 \caption{Illustration of the construction of a scalar ${(r = 4, \delta = 2, \alpha = 1)}$ LRC for $n = 14, \cM=9$ and ${d_{\min}=4}$.} \label{fig:construction1}
\end{figure}
By Theorem~\ref{thm:dmin}, the minimum distance $d_{\min}$ of this code is at most $4$. By Remark~\ref{rm:linearized property}, any $3$ node erasures correspond to at most $2$ rank erasures and then can be corrected by  $\cC^{\rm{Gab}}$, hence $d_{\min}=4$.
In addition, when a single node fails, it can be repaired by using the data stored on other nodes from the same group.

\end{example}

Next, we illustrate Construction I for a vector LRC.

\begin{example}
We consider a DSS with the following parameters:
$$(\cM,n,r,\delta,\alpha)=(28,15,3,3,4).
$$
By~(\ref{eq:upp_bound}) we have $d_{\min}\leq 5$.  Let $N=\frac{15\cdot 3\cdot 4}{3+3-1}=36$ and $(a_1,\ldots, a_{12}, b_1,\ldots, b_{12}, c_1,\ldots, c_{12})$ be a codeword  of an $[36,28,9]_{q^{36}}$ code $\cC^{\textmd{Gab}}$, which is obtained by encoding $\cM = 28$ symbols over $\F=\F_{q^{36}}$ of the original file. The Gabidulin codeword is then divided into three groups $(a_1,\ldots, a_{12})$, $(b_1,\ldots, b_{12})$, and $(c_1,\ldots, c_{12})$. Encoded symbols in each group are stored on three storage nodes as shown in Fig.~\ref{fig:construction}. In the second stage of encoding, a $[4\times 5,3,3]$ MDS array code over $\F_{q}$  is applied on each local group to obtain $\delta-1 = 2$ parity nodes per local group. The coding scheme is illustrated in Fig.~\ref{fig:construction}.

%

\begin{figure}[h]
 \centering
 \includegraphics[width=1\columnwidth, clip=true, trim=0 0 0 240]{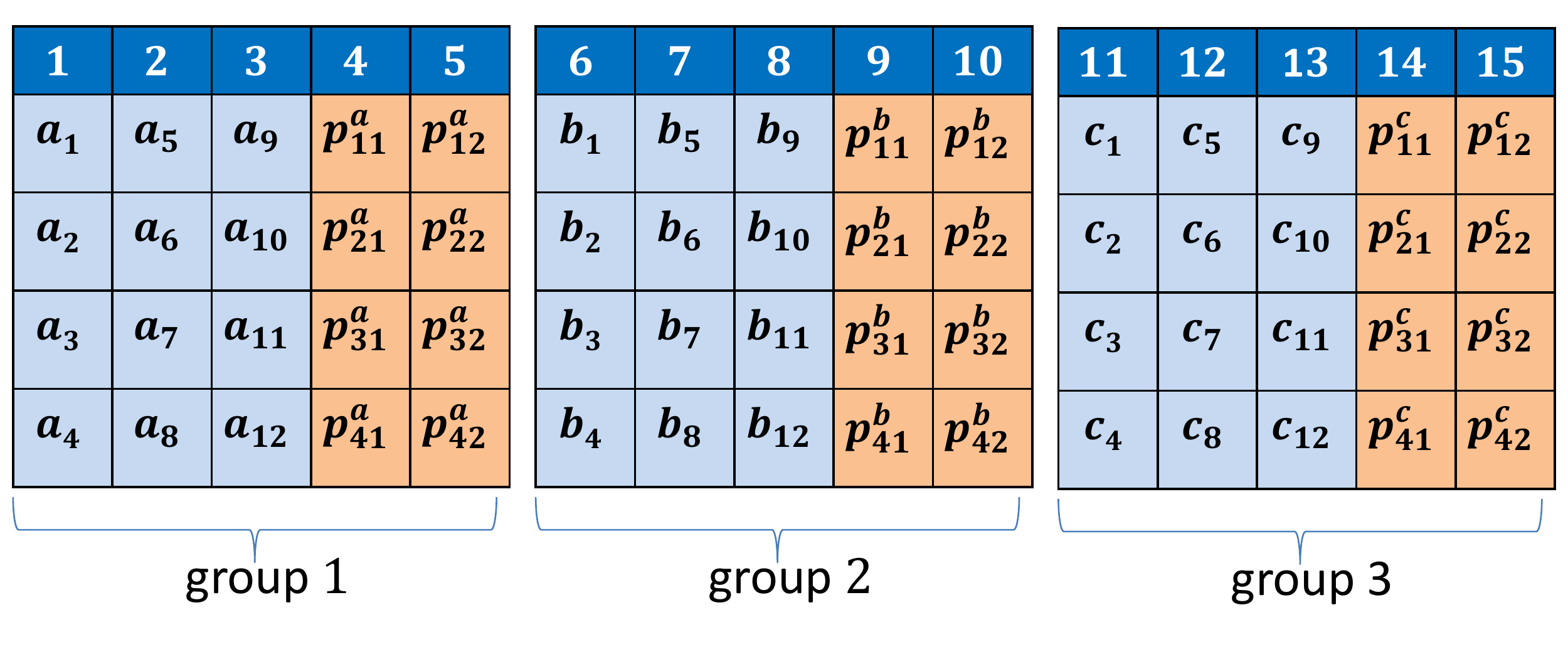}
 \caption{Example of an $(r = 3, \delta = 3, \alpha = 4 )$ locally repairable code with $n = 15$ and $d_{\min}=5$.} \label{fig:construction}
\end{figure}

By Remark~\ref{rm:linearized property}, any $4$ node failures correspond to at most $8$ rank erasures in the corresponding codeword of $\cC^{\textmd{Gab}}$. Since the minimum rank distance of $\cC^{Gab}$ is $9$, these node erasures can be corrected by $\cC^{Gab}$, and thus the minimum distance of $C^{loc}$ is exactly $5$.
\end{example}

\begin{remark}
The efficiency of the decoding of the codes obtained by Construction I depends on the efficiency of the decoding of
the MDS codes and the Gabidulin codes.

\end{remark}

%
%
\section{Repair Efficient LRCs }
\label{sec:discussion}

In this section, we discuss the hybrid codes which for a given locality parameters minimize the repair bandwidth. These codes are based on
 a combination of locally repairable codes with regenerating codes.

In a na\"{\i}ve repair process for a locally repairable code, a newcomer contacts $r$ nodes in its local group and downloads all the data stored on these nodes. 
 Following the line of work of bandwidth efficient repair in DSS given by~\cite{dimakis}, we allow a newcomer to contact $d\geq r$ nodes  in its local group and to download only  $\beta\leq\alpha$ symbols stored in these nodes in order to repair the failed node. The motivation behind this is to lower the repair bandwidth for a LRC. So the idea here is to apply a regenerating code in each local group. (We note that, in a parallel and independent work, Kamath et al. \cite{Kumar12} also proposed utilizing regenerating codes in the context of LRCs.)

In particular, by applying an $(r+\delta-1,r,d,\alpha,\beta)$ MSR code in each local group instead of an MDS array code in the second step of Construction I we obtain a code, denoted by MSR-LRC, which has the maximal minimum distance
(since an MSR code is also an MDS array code), the local minimum storage per node, and the minimized repair bandwidth. (The details of this construction can be found in~\cite{RKSV12}.)

In addition, the optimal scalar codes obtained by Construction I can be used for construction of MBR-LRCs (codes with an MBR code in each local group) as it has been shown by Kamath et al. \cite{Kumar12}.

%


\section{Conclusion}
\label{sec:conclusions}
We presented a novel construction for (scalar and vector) locally repairable codes. This construction is based on maximum rank distance codes. We derived an upper bound on minimum distance for vector LRCs and proved that our construction provides optimal codes for both scalar and vector cases.
We also discussed  how the codes obtained by this construction can be used to construct repair bandwidth efficient  LRCs.




\appendices
\section{Proof of Theorem~\ref{trm:vector_optimality}}
\label{ap:appendix}
To prove that $C^{\rm{loc}}$ attains the bound~(\ref{eq:upp_bound}) we need to show that any $E\triangleq n - \ceilb{\frac{\mathcal{M}}{\alpha}}  - \left(\ceilb{\frac{\mathcal{M}}{r\alpha}}-1\right)(\delta - 1)$ node erasures can be corrected by $C^{\rm{loc}}$. For this purpose we will prove that any $E$ erasures of $C^{\rm{loc}}$  correspond to at most $D-1$ rank erasures of the underlying Gabidulin code and thus can be corrected by the code $\cC^{\rm{Gab}}$. Here, we point out the the worst case erasure pattern is when the erasures appear in the smallest possible number of groups and the number of erasures inside a local group is maximal.


Let $\alpha_0, \alpha_1, \beta_1, \gamma_1$ be the integers such that $N=\alpha (\alpha_0 r+\beta_0)$, $\cM=\alpha(\alpha_1 r+\beta_1)+\gamma_1$, $\alpha_0\geq \alpha_1$, $0\leq \beta_0,\beta_1\leq r-1$, and $0\leq \gamma_1\leq \alpha-1$. Then
\begin{equation}
\label{eq:rank_erasures}
D-1=N-\cM=r \alpha (\alpha_0-\alpha_1)+\alpha(\beta_0-\beta_1)-\gamma_1.
\end{equation}
Then, given $n=\frac{N}{\alpha}+\left\lceil\frac{N}{r \alpha}\right\rceil(\delta-1)$, we can rewrite the bound~(\ref{eq:upp_bound}) in the following way.

\begin{equation}
\label{eq:bound_rewrite}
d_{\min}-1\leq\frac{N}{\alpha}-\left\lceil\frac{M}{\alpha}\right\rceil+\left(\left\lceil\frac{N}{r \alpha}\right\rceil-\left\lceil\frac{M}{r \alpha}\right\rceil+1\right)(\delta-1).
\end{equation}
\begin{enumerate}
  \item Let $(r\alpha)|N$. Then $\beta_0=0$ and $\left\lceil\frac{N}{r \alpha}\right\rceil=\alpha_0$.
  \begin{itemize}
    \item If $\gamma_1=\beta_1=0$ then $\left\lceil\frac{M}{\alpha}\right\rceil=\alpha_1 r$ and $\left\lceil\frac{M}{r \alpha}\right\rceil=\alpha_1$. In this case by~(\ref{eq:bound_rewrite}) we have $d_{\min}-1\leq (r+\delta-1)(\alpha_0-\alpha_1)+\delta-1$. Hence, in the worst case we have $(\alpha_0-\alpha_1)$ groups with all the erased nodes and one additional group with $\delta-1$ erased nodes, which by Remark~\ref{rm:linearized property} corresponds to $r \alpha$ rank erasures in $(\alpha_0-\alpha_1)$ groups of the corresponding Gabidulin codeword. Since by~(\ref{eq:rank_erasures}), $D-1=r \alpha(\alpha_0-\alpha_1)$, this erasures can be corrected by the Gabidulin code.

    \item If $\gamma_1=0$, $\beta_1>0$ then $\left\lceil\frac{M}{\alpha}\right\rceil=\alpha_1 r+\beta_1$ and $\left\lceil\frac{M}{r \alpha}\right\rceil=\alpha_1+1$. Then by~(\ref{eq:bound_rewrite}) we have $d_{\min}-1\leq (r+\delta-1)(\alpha_0-\alpha_1-1)+(r-\beta_1+\delta-1)$. Hence, in the worst case we have $(\alpha_0-\alpha_1-1)$ groups with all the erased nodes and one additional group with $r-\beta_1+\delta-1$ erased nodes, which by Remark~\ref{rm:linearized property} corresponds to $r \alpha (\alpha_0-\alpha_1)-\alpha \beta_1=D-1$ rank erasures that can be corrected by the Gabidulin code.

    \item If $\gamma_1>0$ then $\left\lceil\frac{M}{\alpha}\right\rceil=\alpha_1 r+\beta_1+1$ and $\left\lceil\frac{M}{r \alpha}\right\rceil=\alpha_1+1$.
     Then by~(\ref{eq:bound_rewrite}) we have $d_{\min}-1\leq (r+\delta-1)(\alpha_0-\alpha_1-1)+(r-\beta_1-1+\delta-1)$. Hence, in the worst case we have $(\alpha_0-\alpha_1-1)$ groups with all the erased nodes and one additional group with $r-\beta_1-1+\delta-1$ erased nodes, which corresponds to $r \alpha (\alpha_0-\alpha_1)-\alpha \beta_1-\alpha<D-1$ rank erasures that can be corrected by the Gabidulin code.
  \end{itemize}

  \item Let $\frac{N}{\alpha}(\textmd{mod } r)\geq \frac{\cM}{\alpha}(\textmd{mod } r)>0 $. Then, $\beta_0\geq \beta_1>0$ and $\left\lceil\frac{N}{r \alpha}\right\rceil=\alpha_0+1$.

  \begin{itemize}
       \item If $\gamma_1=0$ then $\left\lceil\frac{M}{\alpha}\right\rceil=\alpha_1 r+\beta_1$ and $\left\lceil\frac{M}{r \alpha}\right\rceil=\alpha_1+1$. Then by~(\ref{eq:bound_rewrite}), we have $d_{\min}-1\leq (r+\delta-1)(\alpha_0-\alpha_1)+(\beta_0-\beta_1+\delta-1)$. Hence, in the worst case we have $(\alpha_0-\alpha_1)$ groups with all the erased nodes and one additional group with $\beta_0-\beta_1+\delta-1$ erased nodes (or $\beta_0+\delta-1$ erased nodes in the smallest group, $(\alpha_0-\alpha_1-1)$ groups with all the erased nodes and one group with $r-\beta_1+\delta-1$ erased nodes). This by Remark~\ref{rm:linearized property} corresponds to $r \alpha (\alpha_0-\alpha_1)-\alpha(\beta_0- \beta_1)=D-1$ rank erasures that can be corrected by the Gabidulin code.

    \item If $\gamma_1>0$ then $\left\lceil\frac{M}{\alpha}\right\rceil=\alpha_1 r+\beta_1+1$ and $\left\lceil\frac{M}{r \alpha}\right\rceil=\alpha_1+1$.
     Then by~(\ref{eq:bound_rewrite}) we have $d_{\min}-1\leq (r+\delta-1)(\alpha_0-\alpha_1)+(\beta_0-\beta_1-1+\delta-1)$. Hence, in the worst case we have $(\alpha_0-\alpha_1)$ groups with all the erased nodes and one additional group with $\beta_0-\beta_1-1+\delta-1$ erased nodes (or $\beta_0+\delta-1$ erased nodes in the smallest group, $(\alpha_0-\alpha_1-1)$ groups with all the erased nodes and one group with $r-\beta_1-1+\delta-1$ erased nodes). This by Remark~\ref{rm:linearized property} corresponds to $r \alpha (\alpha_0-\alpha_1)-\alpha(\beta_0- \beta_1)-\alpha<D-1$ rank erasures that can be corrected by the Gabidulin code.
  \end{itemize}
\end{enumerate}

\end{document}